\documentclass[10pt]{article}
\usepackage{amsmath,amsthm,amsfonts,amssymb,epic,epsfig,graphicx}
\usepackage[dvips]{geometry}
\usepackage{verbatim}
\usepackage{color}

\usepackage{curves}
\usepackage{graphics}
\usepackage{verbatim}

\newcommand{\reals}{\mathbb{R}}

\newcommand{\exponent}{\operatorname{e}}

\theoremstyle{plain}

\newtheorem{stel}{Theorem}

\newtheorem{lemma}{Lemma}
\theoremstyle{remark}

\begin{document}
\title{Failure of antibiotic treatment in microbial populations}

\author{Patrick De Leenheer\footnote{Department of Mathematics, University of Florida, email: deleenhe@math.ufl.edu. Supported by NSF grant DMS-0614651} and Nick Cogan\footnote{Department of Mathematics, Florida State University, email: cogan@math.fsu.edu.}}

\date{}
\maketitle
\begin{abstract}
The tolerance of bacterial populations to biocidal or antibiotic treatment has been well documented in both biofilm and planktonic settings. However, there is still very little known about the mechanisms that produce this tolerance. Evidence that small, non-mutant  subpopulations of bacteria are not affected by antibiotic challenge has been accumulating and provides an attractive explanation for the failure of typical dosing protocols.  Although a dosing challenge can kill all the susceptible bacteria, the remaining persister cells can serve as a source of population regrowth.  
We give a robust condition for the failure of a periodic dosing protocol 
for a general chemostat model, which supports the mathematical conclusions and 
simulations of 
an earlier, more specialized batch model. Our condition implies that
the treatment protocol fails globally, 
in the sense that a mixed bacterial population will ultimately persist above a level 
that is independent of the initial composition of the population. 
We also give a sufficient condition for treatment success, at least for initial 
population compositions near the steady state of interest, corresponding to 
bacterial washout. 
Finally, we investigate 
how the speed at which the bacteria are wiped out depends on the duration 
of administration of the antibiotic. 
We find that this dependence is not necessarily monotone, implying that optimal dosing 
does not necessarily correspond to continuous administration of the antibiotic. Thus, 
genuine periodic protocols can be more advantageous in treating a wide variety of 
bacterial infections.

\end{abstract}

\section{Introduction}   

The failure of antibiotic treatments to eliminate bacterial infections has become both more evident and better understood in the past several decades. Although there is evidence that the over use of antibiotics has amplified the number of chromosomal-resistant bacteria \cite{resistance_chromo}, it is becoming increasingly clear that there are other mechanisms that protect populations of bacteria. Many of these mechanisms depend on whether the bacteria exist in a biofilm or not \cite{cf, davies, keren, desai}.  In particular, the notion that small sub-populations of bacteria may display innate tolerance to various biocides has been proposed as a possible reason for the failure of treatment for bacterial infections \cite{Lewis, balaban, klapper_senescence, cogan}.  Bacteria within a biofilm are enmeshed in a physical gel that provides a secondary boundary that may allow small numbers of bacteria to evade the 
antibiotic; therefore, the failure to eliminate the entire population can allow the population to regrow.

It should be noted that populations of  planktonic bacteria also contain these highly tolerant or persister cells \cite{desai, lewis_planktonic}.  Thus understanding the process of persister formation and the response of the population to biocidal application is fundamental to developing dosing protocols and treatments in both batch culture and biofilm populations.  

As in many areas of biology, mathematical modeling has been used as a counterpart to experimental observations. 
Because there are several hypotheses regarding the mechanism of persister formation, mathematical modeling can be used to provide insight into the success of failure of treatment protocols as well as the consistency of various hypotheses.  Currently, there are at least two distinct hypotheses concerning persister formation-one of which is senescence.  In this case, persister cells are assumed to be those that have undergone many division cycles. It is known that asymmetric division leads to degradation of parts of the cellular machinery that may be the underlying cause of persistence \cite{senscence1}.  Mathematical analysis of a model of senescence has been described in both chemostat and biofilm settings \cite{klapper_senescence, klapper_senescence2}.

A different hypothesis argues that persisters are a phenotype that is expressed at a rate that depends on the external environment. Although the biological details are not well understood, it is thought that this might be due to toxin/antitoxin interaction or some other stress response \cite{Lewis, keren, balaban, lewis_planktonic,lewis_ta}.  This has been investigated mathematically as well \cite{balaban, cogan, cogan_ta, imran}.  In \cite{cogan}, a very simple model of persister formation was developed and optimal dosing protocols, that entail alternating application and resting, were described.  In \cite{cogan_ta}, toxin/antitoxin interaction was explicitly included and the resulting model was analyzed in a chemostat.  In both of these investigations only one particular form of the kinetics was analyzed. One of the goals of the current investigation is to extend these results to a more general form for  kinetics. This is an important process, because it has been shown that not only is there a successful dosing protocol, there is an optimal protocol. Without knowing how robust the model is, there is no way to determine how robust the conclusions are.  Here we take a model of the form proposed in \cite{cogan} and give 
a condition that yields successful dosing protocols, at least for some  
initial makeups of the population.

The manuscript is organized as follows: We begin by describing the model for the dynamics of the bacterial population in response to antibiotic challenge. We then develop the theory by analyzing two extreme  cases (no dosing and constant dosing) and the intermediate case.  
This leads to a local sufficient condition for treatment success. 
Next we give a condition for global treatment failure, 
supported by 
numerical simulations of the model. We also show that the speed of eradication does 
not necessarily depend monotonically on the duration of the administration of the 
antibiotic.

\section{Model}
Consider the following chemostat model:
\begin{eqnarray}
{\dot B_s}&=&\left[\left(1-k_d(t)-k_l(t)\right)f(S)-D\right]B_s+k_g(t)B_p\label{s1}\\
{\dot B_p}&=&k_l(t)f(S)B_s-[k_g(t)+D]B_p\label{s2}\\
{\dot S}&=&D(S^0-S)-\frac{f(S)B_s}{Y}\label{s3}
\end{eqnarray}
where $B_s$ is the concentration of the cells which are susceptible to antibiotics, $B_p$ is the concentration 
of the persister cells which are not affected by the antibiotic and $S$ is the concentration of the nutrient. 
This model deviates from the one in \cite{cogan} because it is a chemostat model, which 
is reflected in the additional loss terms at rate $D$ (called the dilution rate or washout rate), and 
the inflow (at the same rate $D$) of nutrient with an input concentration $S^0$. 
The per capita growth rate of the susceptible cells is denoted by $f(S)$, for which we assume 
the following throughout the rest of the paper:
$$
f:\reals_+\rightarrow \reals_+\textrm{ is smooth and increasing and } f(0)=0.
$$
The persister cells do not consume nutrient, hence the lack of 
a corresponding growth term in the $B_p$-equation. The conversion of nutrient into new biomass occurs with 
a yield of $Y\in (0,1)$. 

The remaining functions 
$k_d(t),k_l(t)$ and $k_g(t)$ are non-negative, time-varying functions which 
describe the effect of antibiotics on the population. First, $k_d(t)f(S)$ is the 
killing rate of the susceptible population. 
Note in particular that the 
killing rate is proportional to the growth rate of the cells. It is positive when both antibiotic and 
nutrient are present, but zero when either one is missing.
Secondly, $k_l(t)$ is the rate at which susceptible cells switch to persister cells. It is also positive 
when antibiotic is present, but zero when this is not the case. Finally, $k_g(t)$ is the rate at which persister 
cells revert to the susceptible state when antibiotic is absent (and zero when antibiotic is present).

Since antibiotics are administered to the reactor vessel in a controlled (lab) environment, 
we will make the simplifying assumption that the functions $k_d(t),k_l(t)$ and $k_g(t)$ 
are $\tau$-periodic (for some given $\tau>0$), and of the bang-bang type with simultaneous switching instances: 
For some $p\in[0,1]$, and for positive parameters $k_d,k_l$ and $k_g$, there holds that
\begin{equation}\label{forms} 
k_d(t)=\begin{cases}k_d\textrm{ for } t\in[0,p\tau)\\0\textrm{ for }t\in [p\tau,\tau)\end{cases},\;\;
k_l(t)=\begin{cases}k_l\textrm{ for } t\in[0,p\tau)\\0\textrm{ for }t\in [p\tau,\tau)\end{cases},\;\;\textrm{ and }
k_g(t)=\begin{cases}0\textrm{ for } t\in[0,p\tau)\\k_g\textrm{ for }t\in [p\tau,\tau)\end{cases}.
\end{equation}
Thus, antibiotics are present during a fraction $p$ of the period $\tau$, and absent during the remaining 
fraction $1-p$ of the period. 

Clearly, this is a simplification of reality because the concentration of an 
antibiotic is not expected to be of the bang-bang type.  
In a more realistic model, the functions $k_d(t), k_l(t)$ and 
$k_d(t)$ would be replaced by functions depending on (at least) a new state variable for the concentration 
of the antibiotic, and the periodicity would arise through a periodic forcing term in the equation for this 
new variable. We leave the study of such a model to the future.

Throughout the rest of this paper we assume that the net effect of the antibiotic alone is lethal to the 
susceptible population:
\begin{equation}\label{lethal}
1-k_d-k_l<0.
\end{equation}
Note that this assumption is valid for the parameter values related to the experiments described in \cite{cogan}.

The main purpose of this paper is to investigate how the behavior of 
system $(\ref{s1})-(\ref{s3})$ with $(\ref{forms})$ changes qualitatively in terms of $p$.

\section{Preliminary results}
In this section we collect a couple of basic results concerning the dynamical behavior of 
$(\ref{s1})-(\ref{s3})$ with $(\ref{forms})$. For a real-valued function $x(t)$, we denote 
the extended real numbers $\liminf_{t\rightarrow \infty}x(t)$ and 
$\limsup_{t\rightarrow \infty}x(t)$ by $x_{\infty}$ and $x^{\infty}$ respectively.

\begin{lemma}\label{dissipative}
System $(\ref{s1})-(\ref{s3})$ with $(\ref{forms})$ has $\reals^3_+$ as a forward invariant set, and it is dissipative.
\end{lemma}
\begin{proof}
The first assertion is obvious. Dissipativity follows by considering the dynamics of 
$$
M=B_s+B_p+YS,
$$
given by 
$$
{\dot M}=D(YS^0-M)-k_d(t)f(S)B_s\leq D(YS^0-M),
$$
and hence 
$$
M^{\infty}\leq YS^0.
$$
\end{proof}
Not only are all state components of every solution ultimately bounded from above by 
some constant which does not depend on initial conditions, we also 
notice in the following Lemma, that $S(t)$ is ultimately bounded from below by some positive constant which is 
independent of initial conditions.

\begin{lemma}\label{lower}
There is a constant $\theta>0$ such that for all solutions of $(\ref{s1})-(\ref{s3})$ with $(\ref{forms})$, 
there holds that $S_{\infty}\geq \theta$.
\end{lemma}
\begin{proof}
Consider the function $g(x):=f(x)S^0-D(S^0-x)$. Then $g$ is increasing with $g(0)<0$ and $g(S^0)>0$, hence 
by the intermediate value theorem, there is a unique $\theta\in (0,S^0)$ such that $g(\theta)=0$. We will show that 
$S_{\infty}\geq \theta$. If not, then since $B_s^{\infty}\leq YS^0$ by the proof of Lemma $\ref{dissipative}$, it follows 
from Corollary $2.4$ in \cite{thieme} -a consequence of the famous Fluctuation Lemma- applied to $(\ref{s2})$ that
\begin{eqnarray*}
0&\geq& \liminf_{t\rightarrow \infty} \left[ D(S^0-S_{\infty})-\frac{f(S_{\infty})B_s(t)}{Y} \right]\\
&\geq & D(S^0-S_{\infty})-f(S_{\infty})S^0\\
&>&D(S^0-\theta)-f(\theta)S^0,
\end{eqnarray*}
which contradicts that $g(\theta)=0$.
\end{proof}

\section{Analysis of the extreme cases $p=0$ and $p=1$.}
First we study the cases where antibiotic is either present or absent for all times. 
Our conclusions are the expected ones: 
When antibiotic is present continuously, all susceptible cells are killed, and consequently the persisters die out as 
well since they can only survive if susceptible cells become persisters. 
When the population is never exposed to antibiotics, 
then all persisters revert to the susceptible state, and ultimately the population will consist entirely of 
susceptible cells (provided the dilution rate is not too high).
\begin{lemma}\label{extreme}
\begin{enumerate}
\item Case $p=1$ ({\bf continuous antibiotic dosing}). The steady state $(B_s,B_p,S)=(0,0,S^0)$ 
of system $(\ref{s1})-(\ref{s3})$ with $(\ref{forms})$ is globally asymptotically stable.
\item Case $p=0$ ({\bf never antibiotic dosing}). If $D<f(S^0)$, then all solutions $(B_s(t),B_p(t),S(t))$ 
with $B_s(0)+B_p(0)>0$ converge to the steady state $(B_s,B_p,S)=(B_s^*,0,S^*)$, where 
$S^*\in (0,S^0)$ is the unique positive value satisfying $f(S^*)=D$, and $B_s^*=(S^0-S^*)Y$.
\end{enumerate}
\end{lemma}
\begin{proof}
\begin{enumerate}
\item
If $p=1$, then $k_d(t)\equiv k_d$, $k_l(t)\equiv k_l$ and $k_g(t)\equiv 0$ for all $t$. By $(\ref{lethal})$ 
it follows that for all solutions,  
${\dot B_s}\leq -DB_s$, and hence $B_s(t)\rightarrow 0$ as $t\rightarrow \infty$. This suggests that 
we should study the linear limiting system 
\begin{eqnarray*}
{\dot B_p}&=&-DB_p\\
{\dot S}&=&D(S^0-S)
\end{eqnarray*}
whose solutions clearly converge to $(B_p,S)=(0,S^0)$. The conclusion now follows 
immediately by applying Theorem F.1 in \cite{smith-waltman}.
\item
If $p=0$, then $k_d(t)=k_l(t)\equiv 0$ and $k_g(t)=k_g$ for all $t$. 
Then the restriction on the initial condition implies that $B_s(t)>0$ for all $t>0$. 
Notice that $B_p(t)\rightarrow 0$ as $t\rightarrow \infty$, suggesting we should study the limiting system
\begin{eqnarray*}
{\dot B_s}&=&[f(S)-D]B_s\\
{\dot S}&=&D(S^0-S)-\frac{f(S)B_s}{Y}
\end{eqnarray*}
This is the classical chemostat model of a single cell population growing on a single nutrient, described in 
\cite{smith-waltman}: all solutions with $B_s(0)>0$ converge to $(B_s,S)=(B_s^*,S^*)$ under the assumptions of the 
Lemma. The conclusion follows again by Theorem F.1 in \cite{smith-waltman}.
\end{enumerate}
\end{proof}

\section{The case of periodic dosing: $p\in (0,1)$.}
In this section we deal with the $\tau$-periodic model $(\ref{s1})-(\ref{s3})$ with $(\ref{forms})$, 
assuming that $p\in(0,1)$. 
We will also assume that the dilution rate is not too large, as in Lemma $\ref{extreme}$:
\begin{equation}\label{small}
D<f(S^0).
\end{equation}

Notice that system $(\ref{s1})-(\ref{s3})$ with $(\ref{forms})$ 
has a steady state $E_0:=(B_s,B_p,S)=(0,0,S_0)$, regardless of the value of $p$.
To determine its stability properties, we determine 
the $\tau$-periodic variational equation at $E_0$:
$$
{\dot z}=\begin{pmatrix}(1-k_d(t)-k_l(t))f(S^0)-D&k_g(t)&0\\f(S^0)k_l(t)&-(k_g(t)+D)&0\\-\frac{f(S^0)}{Y}&0&-D\end{pmatrix}z
$$
and the Floquet multipliers of this system are $\exponent^{-D\tau}$ (which is of course inside the unit circle of 
the complex plane) and the Floquet multipliers of
$$
{\dot x}=\begin{pmatrix}(1-k_d(t)-k_l(t))f(S^0)-D&k_g(t)\\f(S^0)k_l(t)&-(k_g(t)+D)\end{pmatrix}x
$$
Using $(\ref{forms})$, these Floquet multipliers are the eigenvalues of the following matrix:
\begin{equation}\label{transition}
\Phi:=\exponent^{(1-p)\tau A_2}\exponent^{p\tau A_1},
\end{equation}
where
\begin{equation}\label{matrices}
A_1:=\begin{pmatrix}(1-k_d-k_l)f(S^0)-D&0\\k_lf(S^0)&-D\end{pmatrix},\textrm{ and } 
A_2:=\begin{pmatrix}f(S^0)-D&k_g\\0&-(k_g+D) \end{pmatrix}
\end{equation}
are triangular, quasimonotone matrices. Notice also that $A_1$ is Hurwitz by $(\ref{lethal})$, 
and that $A_2$ has one negative and one positive eigenvalue by $(\ref{small})$.

Since $A_1$ and $A_2$ are quasi-monotone (i.e. their off-diagonal entries are all non-negative), it follows that 
their matrix exponentials are (entry-wise) non-negative, triangular matrices 
and then their product $\Phi$ is a (entry-wise) positive matrix whose spectral radius $\rho(\Phi)$ is an eigenvalue 
by the Perron-Frobenius Theorem\cite{berman-plemmons}. 
Consequently, to determine stability of $E_0$, we need to establish 
whether or not $\rho(\Phi)$ is inside the unit circle: If $\rho(\Phi)<1$, then $E_0$ is locally 
asymptotically stable. If $\rho(\Phi)>1$, then $E_0$ is unstable.
Summarizing, we have established 
\begin{stel}\label{loc-stab}
Let $p\in (0,1)$, and assume that $(\ref{small})$ holds. Then the steady state $E_0=(0,0,S^0)$ is locally 
stable for $(\ref{s1})-(\ref{s3})$ with $(\ref{forms})$ if $\rho(\Phi)<1$, but unstable if $\rho(\Phi)>1$.
\end{stel}

Our main concern is knowing how $\rho(\Phi)$ varies as a continuous function of $p$ (this variation is continuous 
since eigenvalues of a matrix are continuous functions of its entries, and clearly the entries of $\Phi$ 
are continuous in $p$).  
For $p=0$ (never using antibiotic), and hence also for $p$ near $0$ by continuity of $\rho(\Phi)$, 
we have that $\rho(\Phi)=\rho(\exponent^{\tau A_2})=\exponent^{(f(S^0)-D)\tau}>1$. This is in accordance with 
Lemma $\ref{extreme}$, where it was shown that all solutions with $B_s(0)+B_p(0)>0$ 
converge to $(B_s^*,0,S^*)$, and thus $E_0=(0,0,S^0)$ must be unstable. 
For $p=1$ (using antibiotic continuously), and hence also for $p$ near $1$, we have 
that $\rho(\Phi)=\rho(\exponent^{\tau A_1})=\exponent^{-D\tau}<1$. This is in accordance with 
Lemma $\ref{extreme}$ as well because it was shown there that all solutions converge to $E_0$ in this case.

We can actually determine the dependence of $\rho(\Phi)$ on $p$ explicitely because 
fortunately, both $A_1$ and $A_2$ are diagonalizable (their eigenvalues are distinct), which 
simplifies the computation of their matrix exponentials somewhat. It is easily verified that 
$$
A_1=T_1D_1T_1^{-1}\textrm{ and } A_2=T_2D_2T_2^{-1},
$$
where 
$$
T_1=\begin{pmatrix}1-k_d-k_l&0\\k_l&1\end{pmatrix},\;\; D_1=\begin{pmatrix}(1-k_d-k_l)f(S^0)-D&0\\
0&-D \end{pmatrix}
$$
and
$$ 
T_2=\begin{pmatrix}1&k_g\\0&-(f(S^0)+k_g)\end{pmatrix},\;\; D_2=\begin{pmatrix}f(S^0)-D&0\\
0&-(k_g+D) \end{pmatrix},
$$
and thus using $(\ref{transition})$ that
$$
\Phi=T_2\exponent^{(1-p)\tau D_2}T_2^{-1}T_1\exponent^{p\tau D_1}T_1^{-1}.
$$
A lengthy algebraic calculation shows that the positive matrix 
$\Phi=\begin{pmatrix}\Phi_{11}&\Phi_{12}\\ \Phi_{21}&\Phi_{22}\end{pmatrix}$ is given by
\begin{eqnarray*}
c\Phi_{11}&=&-(1-k_d-k_l)(f(S^0)+k_g)\exponent^{\tau [f(S^0)-D-p(k_d+k_l)f(S^0)]}\\
&&+k_gk_l\left(\exponent^{(1-p)\tau (f(S^0)-D)}-
\exponent^{-(1-p)\tau (k_g+D)} \right)\left(\exponent^{-p\tau D}-\exponent^{p\tau [(1-k_d-k_l)f(S^0)-D]} \right),
\end{eqnarray*}
\begin{equation*}
c\Phi_{12}=
-k_g(1-k_d-k_l)\exponent^{-p\tau D}\left(\exponent^{(1-p)\tau (f(S^0)-D)}-\exponent^{-(1-p)\tau (k_g+D)} \right),
\end{equation*}
\begin{equation*}
c\Phi_{21}=k_l(f(S^0)+k_g)\exponent^{-(1-p)\tau (k_g+D)}\left(\exponent^{-p \tau D}-\exponent^{p \tau [(1-k_d-k_l)f(S^0)-D]} \right)
\end{equation*}
and 
\begin{equation*}
c\Phi_{22}=-(f(S^0)+k_g)(1-k_d-k_l)\exponent^{-\tau \left((1-p)k_g+D\right)}
\end{equation*}
where 
$$
c=-(1-k_d-k_l)(f(S^0)+k_g)
$$
is a positive constant, independent of $p$.

Since $\Phi$ is a positive matrix, its 
spectral radius can now be calculated explicitly in terms of its entries:
\begin{equation}\label{formule}
\rho(\Phi)=\frac{\Phi_{11}+\Phi_{22}+\sqrt{(\Phi_{11}-\Phi_{22})^2+4\Phi_{12}\Phi_{21}}}{2}.
\end{equation}

\section{Conditions for treatment failure}

In this section we show that the spectral radius $\rho(\Phi)$ also plays a key role in the global behavior of 
system $(\ref{s1})-(\ref{s3})$ with $(\ref{forms})$ and $p\in (0,1)$. 
We will show that if $\rho(\Phi)>1$, then not only is $E_0$ unstable as we have shown in Theorem $\ref{loc-stab}$, but 
treatment fails globally, because both cell populations persist uniformly. 
In addition we will show that there are positive periodic solutions.

\begin{stel}\label{main-persist}
Let $p\in (0,1)$, and assume that $(\ref{small})$ holds. 
If $\rho(\Phi)>1$, then treatment fails and the population is uniformly persistent, i.e. there is 
some $\epsilon^*>0$ (independent of initial conditions), such that all solutions of $(\ref{s1})-(\ref{s3})$ with 
$(\ref{forms})$ and $B_s(0)>0$, have the property that:
$$
B_s(t)>\epsilon^*,\textrm{ and } B_p(t)>\epsilon^*,\;\;\textrm{ for all sufficiently large } t.
$$
Moreover, there are $\tau$-periodic solutions $(B_s(t),B_p(t), S(t))$ with $B_s(t), B_p(t)>0$ for all $t$.
\end{stel}
\begin{proof}
Define the following matrix:
$$
{\tilde \Phi}(\epsilon)=\exponent^{(1-p)\tau {\tilde A}_2(\epsilon)}\exponent^{p\tau {\tilde A}_1(\epsilon)},
$$
where
$$
{\tilde A}_1(\epsilon):=\begin{pmatrix}f(S^0-\epsilon)-(k_d+k_l)f(S^0+\epsilon)-D&0\\k_lf(S^0-\epsilon)&-D\end{pmatrix},
\textrm{ and } 
{\tilde A}_2(\epsilon):=\begin{pmatrix}f(S^0-\epsilon)-D&k_g\\0&-(k_g+D) \end{pmatrix}
$$
Notice that ${\tilde \Phi}(0)=\Phi$, and thus since $\rho(\Phi)>1$, it follows that 
\begin{equation}\label{perturbation}
\rho({\tilde \Phi}(\epsilon))>1,\;\; \textrm{for all sufficiently small } \epsilon>0,
\end{equation}
as well, because the spectral radius of any matrix is continuous with respect to its entries. We fix some $\epsilon>0$ 
such that $(\ref{perturbation})$ holds. 

We will first show that $B_s$ is uniformly weakly persistent, i.e. that there is some 
$\epsilon'>0$ such that if $B_s(0)>0$, then $B_s^\infty\geq \epsilon'$. 
By contradiction, if $B_s$ is not uniformly weakly persistent, then there is some solution $(B_s(t),B_p(t),S(t))$ with 
$B_s(0)>0$ such that 
\begin{equation}\label{bound}
B_s^{\infty}\leq \frac{YD}{2f(S^0)}\epsilon.
\end{equation}

By Corollary $2.4$ in \cite{thieme} applied to $(\ref{s3})$, 
and since $S^{\infty}\leq S^0$ by $(\ref{s3})$, we have that
\begin{eqnarray*}
0&\geq&\liminf_{t\rightarrow \infty}\left[ D(S^0-S_{\infty})- \frac{f(S_{\infty})B_s(t)}{Y} \right]\\
&\geq&D(S^0-S_{\infty})- \frac{f(S^0)B_s^{\infty}}{Y},\\
\end{eqnarray*}
and hence by $(\ref{bound})$ that 
$$
S_{\infty}\geq S^0-\frac{\epsilon}{2}.
$$
Thus, for some $T^*>0$, there holds that $S^0-\epsilon\leq S(t)\leq S^0+\epsilon$ for all $t\geq T^*$. 
It follows from $(\ref{s1})-(\ref{s2})$, that for all $t\geq T^*$:
\begin{equation}\label{compare}
\begin{pmatrix}
{\dot B_s}\\
{\dot B_p}
\end{pmatrix}
\geq
\begin{pmatrix}
f(S^0-\epsilon)-(k_d(t)+k_l(t))f(S^0+\epsilon)-D&k_g(t)\\
k_l(t)f(S^0-\epsilon)&-(k_g(t)+D)
\end{pmatrix}
\begin{pmatrix}
B_s\\
B_p
\end{pmatrix}
\end{equation}
where the vector inequalities should be interpreted componentwise. Notice that the vector field 
on the right-hand side of $(\ref{compare})$ is that of a $\tau$-periodic, cooperative linear system 
whose principal fundamental matrix solution evaluated over one period $\tau$ equals 
${\tilde \Phi}(\epsilon)$. By Kamke's comparison Theorem 
(see e.g. Theorem B.1 in Appendix B of \cite{smith-waltman}) it follows that for all $t\geq T^*$, the vector 
$(B_s(t),B_p(t))^T$ is not smaller (component-wise) than the solution starting in $(B_s(T^*),B_p(T^*))^T$ 
of the $\tau$-periodic, cooperative linear system with vector field given in the right-hand side of $(\ref{compare})$. 
But all non-zero, non-negative solutions of the linear system diverge because $\rho({\tilde \Phi}(\epsilon))>1$. 
Then so does $(B_s(t), B_p(t))$, and this contradicts $(\ref{bound})$.
We have thus shown that $B_s$ is uniformly weakly persistent.

Next we establish that $B_s$ is in fact uniformly strongly persistent. This follows from 
Theorem $1.3.3$ in \cite{zhao}, applied to the map $P$ which maps $(B_s(0),B_p(0),S(0))^T\in X$ to 
$(B_s(\tau), B_p(\tau),S(\tau))^T$, where  
$X:=\{(B_s, B_p, S)^T\in \reals^3_+\; | \; B_s+B_p+YS\leq YS^0\;\}$, $X_0:=\{(B_s, B_p, S)^T\in X\; |\; 
B_s\neq 0\}$ and $\partial X_0:=\{(B_s, B_p, S)^T\in X\; |\; B_s=0\}$. The map $P$ is continuous and maps 
$X_0$ into itself, and it has a global attractor because it is compact and dissipative. It follows that there is some 
$\epsilon_1^*>0$, independent of initial conditions, such that if $B_s(0)>0$, then $\liminf_{n \rightarrow \infty} 
B_s(n\tau)> \epsilon_1^*$, and also that $\liminf_{t\rightarrow \infty}B_s(t)>\epsilon_1^*$ by 
Theorem $3.1.1$ in \cite{zhao}.

Next we show that uniform strong persistence of $B_s$, implies uniform strong persistence of $B_p$. Consider 
equation $(\ref{s2})$ and notice that for all sufficiently large $t$:
$$
{\dot B_p}\geq k_l(t)f(\theta)B_s-(k_g(t)+D)B_p\geq k_l(t)f(\theta)\frac{\epsilon_1^*}{2}-(k_g(t)+D)B_p,
$$
where $\theta$ is the positive constant from Lemma $\ref{lower}$. 
It is not hard to show that the linear equation
$$
{\dot z}=k_l(t)f(\theta)\frac{\epsilon_1^*}{2}-(k_g(t)+D)z,
$$
has a positive $\tau$-periodic solution $p(t)$ and that all non-negative solutions converge to it. Therefore, it follows that 
for all sufficiently large $t$,
$$
B_p(t)\geq \frac{p_{\infty}}{2},
$$
establishing uniform strong persistence for $B_p$, since $p_{\infty}$ is independent of initial conditions. We conclude 
the proof of uniform strong persistence of $B_s$ and $B_p$ by setting $\epsilon^*=\min \{\epsilon_1^*,\frac{p_{\infty}}{2}\}$. 

Finally, to show that there are $\tau$-periodic solutions with $B_s(t),B_p(t)>0$, we apply Theorem $1.3.6$ from \cite{zhao} applied 
to the continuous map $P$ defined above. We have already remarked that this map is continuous, maps 
$X_0$ into itself, is dissipative and compact, and 
we have just proved that it is uniformly strongly 
persistent with respect to $(X_0,\partial X_0)$. Observe also that $X_0$ is relatively open in $X$, 
and that $X_0$ is convex. Then by Theorem $1.3.6$ from \cite{zhao}, the map $P$ has a fixed point in $X_0$, and this 
in turn implies the existence of $\tau$-periodic solutions with $B_s(t)>0$ for $(\ref{s1})-(\ref{s3})$. 
The same argument as the one used above to establish uniform strong persistence of $B_p(t)$, shows that these 
$\tau$-periodic solutions are such that $B_p(t)>0$ as well.

\end{proof}

\section{Numerical example}
We use the following numerical values:
The per capita growth rate is of Michaelis-Menten type:
$$
f(S)=\frac{\mu S}{k_s+S},\textrm{ where } \mu=0.417 \textrm{ hs}^{-1}\textrm{ and } k_s=0.2\textrm{ mgl}^{-1},
$$
and the parameters 
$$
k_d=3,\; k_l=0.1,\; k_g=0.5\textrm{ h}^{-1},\textrm{ and } \tau=10\textrm{ h}.
$$

The chemostat setting requires that we specify two additional parameters:
$$
D=0.1\textrm{ h}^{-1}\textrm{ and } S^0=1\textrm{ mgl}^{-1}.
$$
It can be easily verified that these choices satisfy the conditions $(\ref{lethal})$ and $(\ref{small})$.
Note that there is no need to specify the yield coefficient $Y$ in order to calculate $\rho(\Phi)$. Indeed, the matrix 
$\Phi$ in $(\ref{transition})$ does not depend on $Y$.

The graph of the spectral radius $\rho(\Phi)$ in terms of $p$, determined using formula $(\ref{formule})$, 
is given in Figure $\ref{spec}$. We see that $\rho(\Phi)=1$ when $p$ is approximately equal to $0.242$. 
Clearly, $\rho(\Phi)$ is not monotone, and it has a global mimimum of approximately $0.098$ 
which is achieved at $p\approx 0.612$ (determined numerically using Mathematica). 
Thus, the optimal strategy in a dosing experiment with a period of $\tau=10$ hours 
occurs for a dosing duration of $p\tau\approx 6.12$ hours. Here, optimality means that 
eradication happens as quickly as possible.
\begin{figure}
\centerline{
\includegraphics[width=12cm]{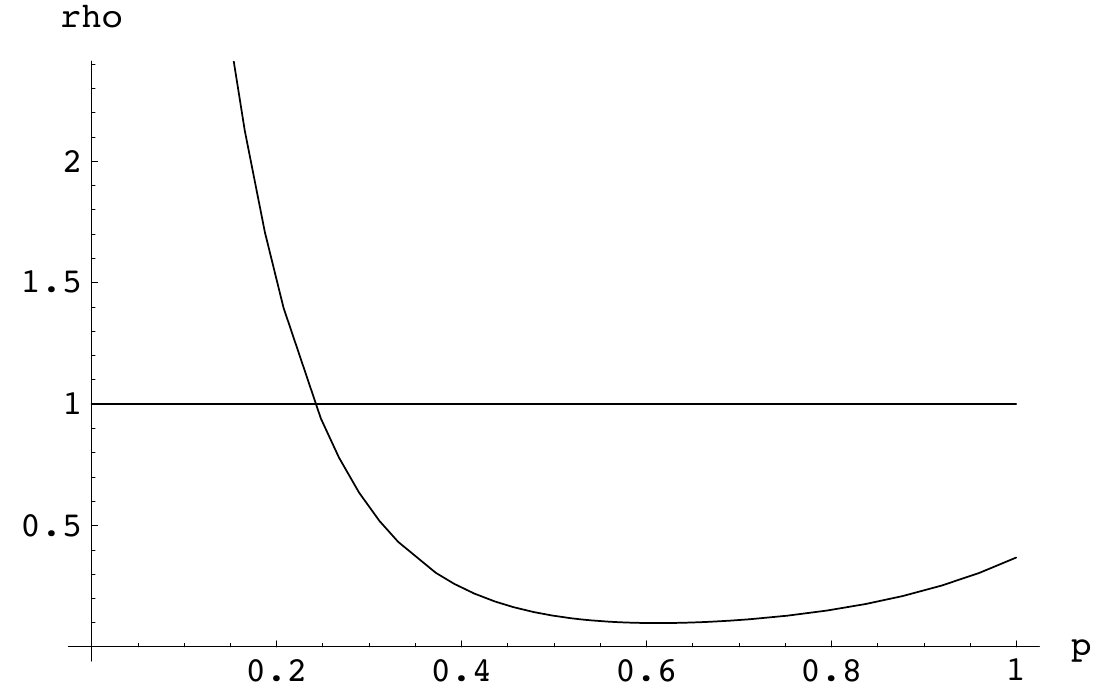}}
\caption{Spectral radius of $\Phi$ for $p\in [0,1]$ ($k_d=3$, other parameters in text.)}\label{spec}
\end{figure}

Let us also illustrate what happens if $p$ equals approximately $0.205$. Then 
$\rho(\Phi)$ equals approximately $1.436$, implying that treatment fails. 
It appears that the solution of $(\ref{s1})-(\ref{s3})$ with $Y=1$ 
(and all other parameters as above) starting from the 
initial condition $(B_s,B_p,S)=(0.3,0,0.4)$ converges to a 
$\tau$-periodic solution, see Figures $\ref{timeseries1}$ and $\ref{timeseries2}$. These observations 
are in accordance with Theorem $\ref{main-persist}$.

\begin{figure}
\centering
\begin{tabular}{cc}
\includegraphics[height=50mm]{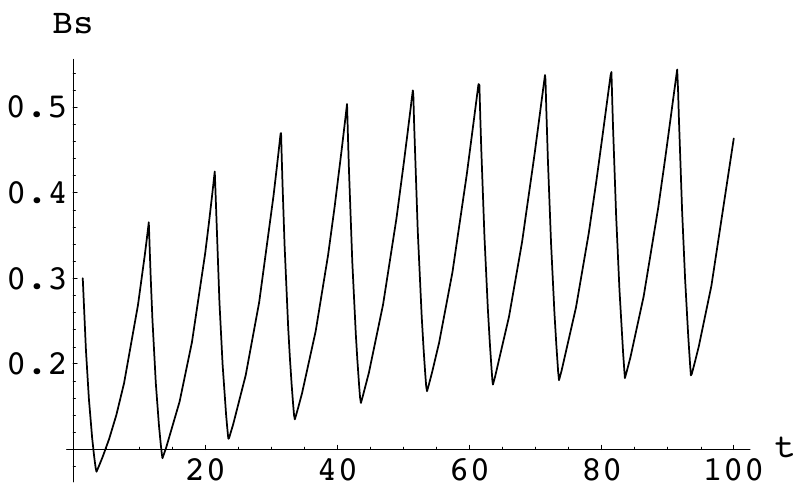}&
\includegraphics[height=50mm]{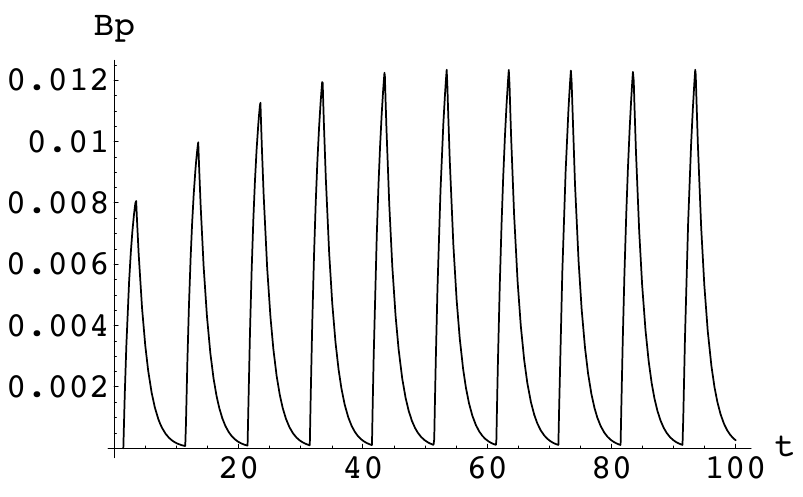}
\end{tabular}
\caption{Times series for $B_s$ and $B_p$. ($p\approx 0.205$, $\rho(\Phi)\approx 1.436$ so treatment 
fails)}\label{timeseries1}
\end{figure}

\begin{figure}
\centering
\begin{tabular}{cc}
\includegraphics[height=50mm]{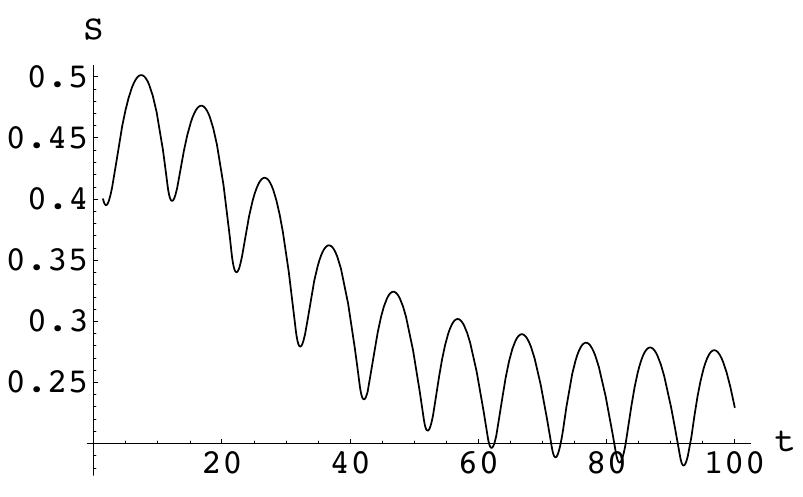}&
\includegraphics[height=50mm]{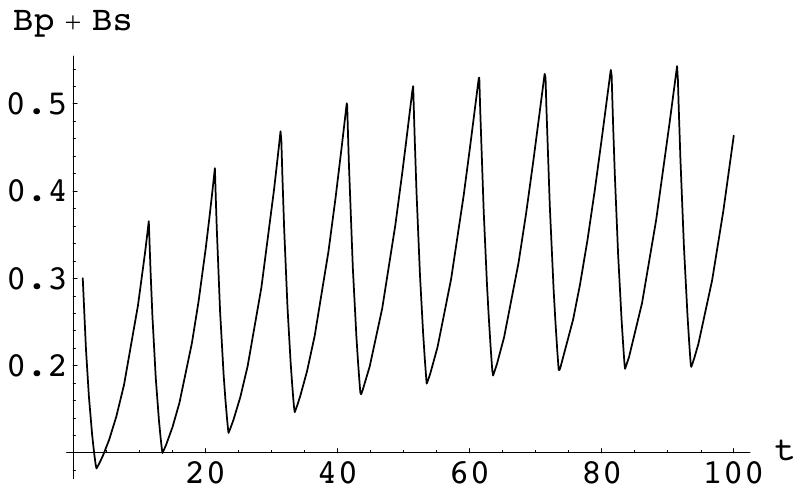}
\end{tabular}
\caption{Times series for $S$ and $B_s+B_p$. ($p\approx 0.205$, $\rho(\Phi)\approx 1.436$ so treatment 
fails)}
\label{timeseries2}
\end{figure}

We also see that if $p=\frac{1}{3}$, then $\rho(\Phi)$ equals approximately $0.426$. 
Then it follows from Theorem $\ref{loc-stab}$ that $E_0$ is locally stable. Figures 
$\ref{timeseries3}$ and $\ref{timeseries4}$ illustrate this by indicating that for the 
solution with the same initial condition $(B_s,B_p,S)=(0.3,0,0.4)$ as above, treatment 
is successful.

\begin{figure}
\centering
\begin{tabular}{cc}
\includegraphics[height=50mm]{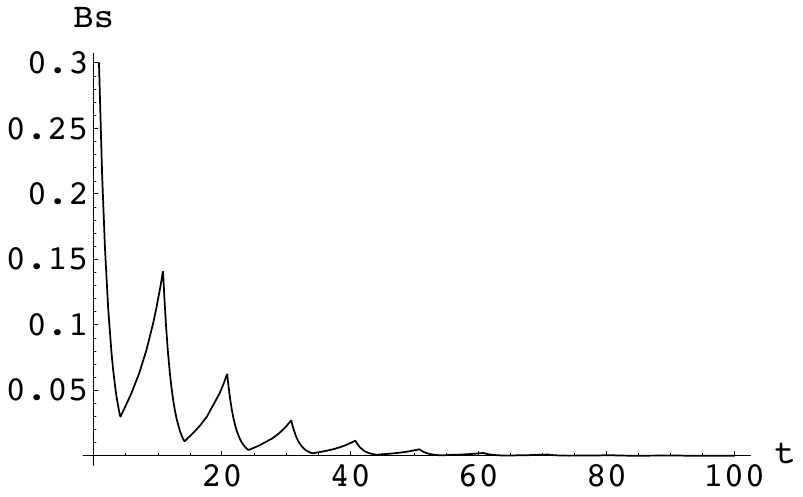}&
\includegraphics[height=50mm]{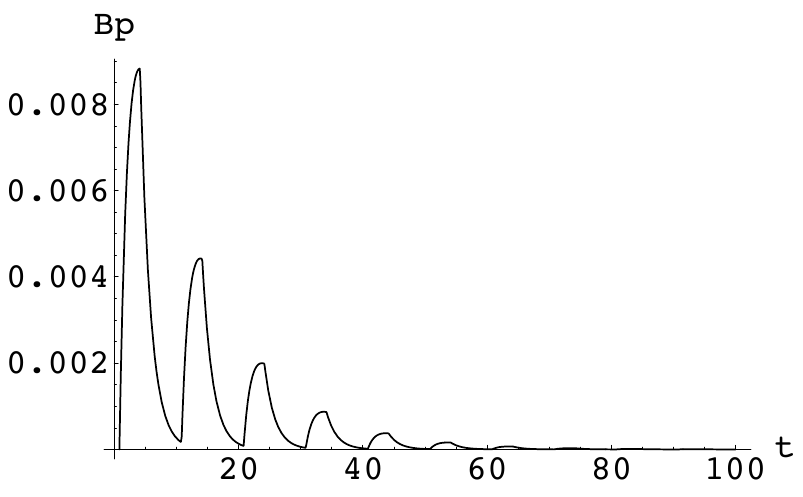}
\end{tabular}
\caption{Times series for $B_s$ and $B_p$. ($p=\frac{1}{3}$, $\rho(\Phi)\approx 0.426$, so treatment 
succeeds)}\label{timeseries3}
\end{figure}

\begin{figure}
\centering
\begin{tabular}{cc}
\includegraphics[height=50mm]{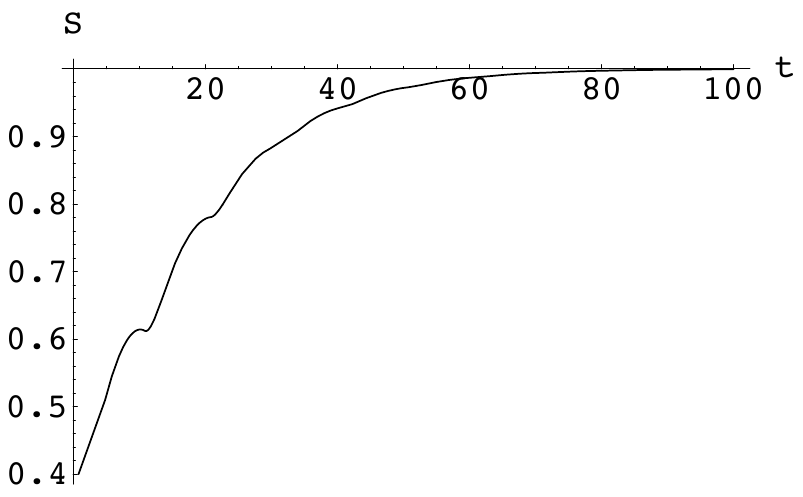}&
\includegraphics[height=50mm]{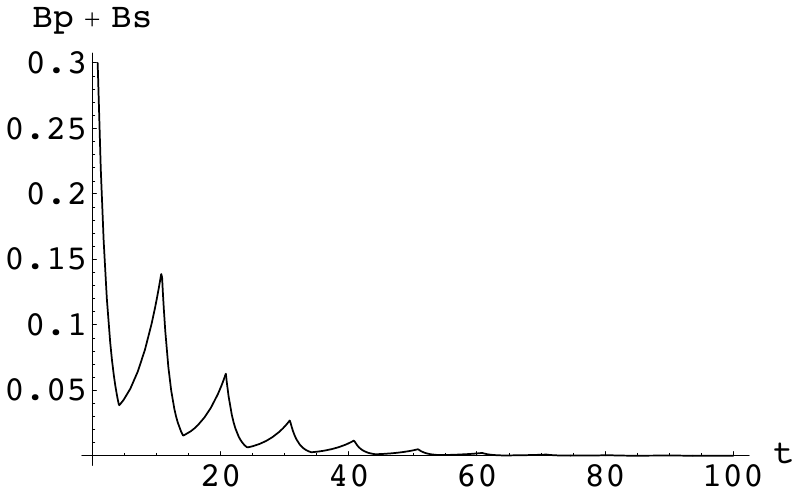}
\end{tabular}
\caption{Times series for $S$ and $B_s+B_p$. ($p=\frac{1}{3}$, $\rho(\Phi)\approx 0.426$, so treatment 
succeeds)}
\label{timeseries4}
\end{figure}

Finally we remark that the spectral radius of $\Phi$ may be monotone: 
If we change the value of $k_d$ from $3$ to $1$, and leave all 
other parameters unchanged, then $\rho(\Phi)$ is a decreasing function of $p\in [0,1]$, 
and achieves its minimum at $p=1$. This indicates 
that for this case, the optimal strategy is to use antibiotics continuously, see Figure $\ref{spec-mon}$.

\begin{figure}
\centerline{
\includegraphics[width=12cm]{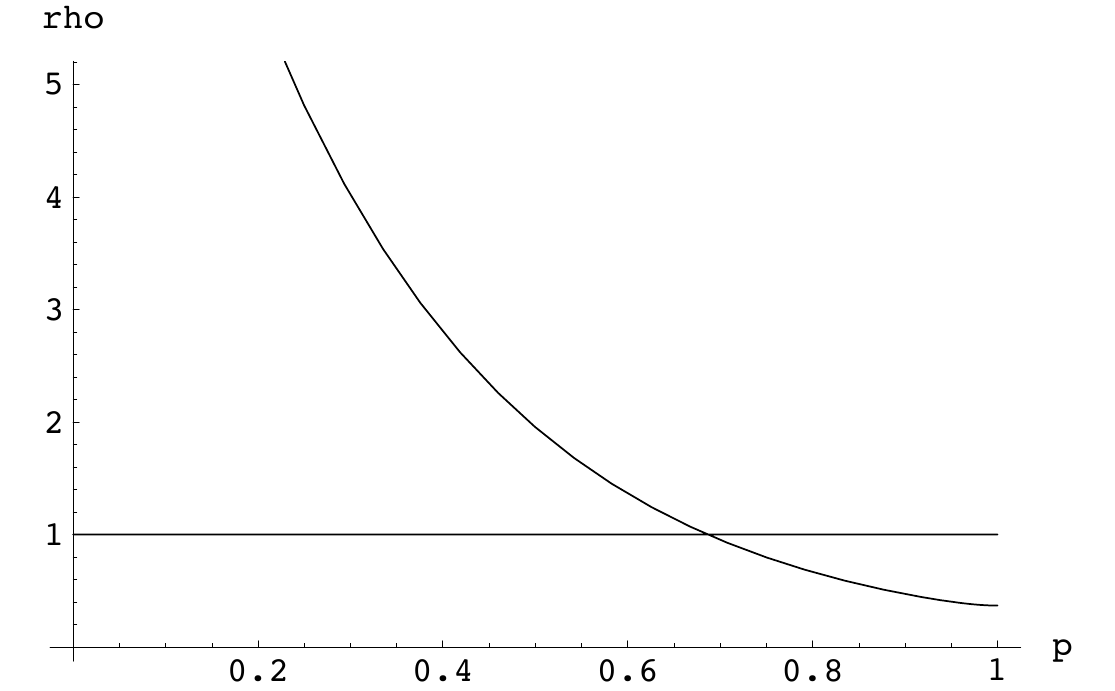}}
\caption{Spectral radius of $\Phi$ for $p\in [0,1]$ ($k_d=1$, other parameters in text).}\label{spec-mon}
\end{figure}

\section{Discussion}

Bacterial infections are a source of problems in a wide variety of situations including industrial, environmental and clinical settings.  Growing understanding of the inability of antibiotics and biocides to treat these infections has driven investigations into the cause of the failure of treatments. It is becoming increasingly evident that persister cells must play an important role in protecting populations of bacteria. It has been observed that other protective mechanisms including physiological and physical processes are not sufficient to explain the observed failures \cite{cogan_phys, chambliss}. Moreover, because it is very difficult to investigate and classify persister cells experimentally, mathematical modeling can play an important role in supporting hypotheses as well as generating useful predictions.

We have described and analyzed a general model for the dynamics of persister formation in response to antibiotic challenge. We have been able to 
provide a condition for the success/failure of antibiotic challenge in a chemostat. 
We have found that there is an optimal strategy, that is one that kills the bacteria the fastest.  These results indicate that periodic dosing is an effective treatment protocol for a variety of bacteria, substantially strengthening the results in \cite{cogan}.  The theoretical results were also confirmed by direct numerical simulations.


\begin{thebibliography}{199}




 

\bibitem{klapper_senescence2}
{\sc   B~.P~. Ayati, I.~Klapper}, {\em A Multiscale Model of Biofilm as a Senescence-Structured Fluid},  SIAM Multi. Model. Sim. 6 (2007) pp.~347-365


\bibitem{balaban}
{\sc N.~Q. Balaban, J.~Merrin, R.~Chait, L.~Kowalik, and S.~Leibler}, {\em
  Bacterial persistence as a phenotypic switch}, Science, 305 (2005),
  pp.~1622--1625.



\bibitem{berman-plemmons} {\sc A. Berman, and R. Plemmons}, {\em Nonnegative matrices in the 
mathematical sciences}, SIAM, 1994. 

\bibitem{chambliss}
{\sc J.~D. Chambliss, S.~M. Hunt, and P.~S. Stewart}, {\em A three-dimensional
  computer model of four hypothetical mechanisms protecting biofilms from
  antimicrobials}, Appl. Environ. Microbiol., 72 (2006), pp.~2005--2013.


\bibitem{cogan} {\sc N. G. Cogan}, {\em Effects of persister formation on bacterial response dosing}, Journal of Theoretical 
Biology 238 (2006), pp.~694-703.

\bibitem{cogan_ta}
{\sc N.~ G.~ Cogan}, {\em Incorporating Toxin Hypothesis into a Mathematical Model of Persister Formation and Dynamics}, Journal of Theoretical Biology 248 (2007): 340-349 

\bibitem{cogan_phys}
{\sc N.~G.~ Cogan, Ricardo Cortez and Lisa J.~ Fauci}, {\em Modeling Physiological Resistence in Bacterial Biofilms}, Bulletin of Mathematical Biology 67 (4) (2005) , pp.~ 831-853 

\bibitem{cf}
{\sc J.~Costerton}, {\em Cystic fibrosis pathogenesis and the role of biofilms
  in persistent infection}, Trends Microbiol., 9 (2001), pp.~50--52.

  
  \bibitem{desai}
{\sc M.~Desai, T.~Buhler, P.~Weller, and M.~Brown}, {\em Increasing resistance
  of planktonic and biofilm cultures of \it{{B}urkholderia cepecia}x to
  ciproflaxacin and ceftazidime during exponential growth}, Journal of
  Antimicrobial Chemotherapy, 42 (1998), pp.~153--160.

\bibitem{davies}
{\sc D.~Davies}, {\em Understanding biofilm resistance to antibacterial
  agents}, Nature Reviews Drug Discovery, 2 (2003), pp.~114--122.
  
  
  \bibitem{resistance_model_2}
{\sc M.~G. Dodds, K.~J. Grobe, and P.~S. Stewart}, {\em Modeling biofilm
  antimicrobial resistance}, Biotechnology and Bioengineering, 68 (2000),
  pp.~456--465.


\bibitem{imran} {\sc M. Imran, and H.L. Smith}, 
{\em The pharmacodynamics of antibiotic treatment}, 
Journal of Computational and Mathematical Methods in Medicine 7(2006), pp.~229 - 263.

 \bibitem{keren}
{\sc I.~Keren, N.~Kaldalu, A.~Spoering, Y.~Wang, and K.~Lewis}, {\em Persister
  cells and tolerance to antimicrobials}, FEMS Microbiology Letters, 230
  (2004), pp.~13--18.


\bibitem{klapper_senescence}
{\sc I.~Klapper, P.~Gilbert, B.~P. Ayati, J.~Dockery, and P.~S. Stewart}, {\em Senescence can explain microbial persistence},  Microbiology, 153(2007) pp.~3623-3630.


\bibitem{phys_resistance}
{\sc H.~M. Lappin-Scott and J.~W. Costerton}, eds., {\em Microbial Biofilms},
  Cambridge University Press, Cambridge, 1995, ch.~Mechanisms of the Protection
  of Bacterial Biofilms from Antimicrobial Agents, pp.~118--130.

\bibitem{Lewis}
{\sc K.~Lewis}, {\em Riddle of biofilm resistance}, Antimicrobial Agents and
  Chemotherapy, 45 (2001), pp.~999--1007.


\bibitem{lewis_ta}
{\sc K.~Lewis}, {\em Persister cells and the riddle of biofilm survival},
  Biochemistry - Moscow, 70 (2005), pp.~267--285.


\bibitem{resistance_chromo}{\sc Harold C. Neu}, {\em The Crisis in Antibiotic Resistance}, Science, 257(5073), 1992.

\bibitem{senscence1}
{\sc E.~J.~Stewart and R.~ Madden and G.~ Paul and F.~ Taddei}, {\em Aging and death in an
organism that reproduces by morphologically symmetric division}, PLoS Biology 3 (2005), pp.~295-300

\bibitem{smith-waltman} {\sc H.L. Smith, and P. Waltman}, {\em The Theory of the 
Chemostat}, Cambridge University Press, 1995.

\bibitem{lewis_planktonic}
{\sc A.~ Spoering and K.~Lewis},{\em Biofilms and Planktonic Cells of Pseudomonas aeruginosa Have Similar Resistance to Killing by Antimicrobials}, Journal of Bacteriology, 183(23) (2001), pp.~6746-6751.


\bibitem{barbara}
{\sc B.~Szomoloy and I.~Klapper and J.~Dockery and P.~Stewart},{\em Adaptive responses to antimicrobial agents in biofilms}, Environmental Microbiology, 7(8) (2005), pp.~1186-1191

\bibitem{thieme} {\sc H.R. Thieme}, {\em Persistence under relaxed point-dissipativity (with applications to an endemic model)}, 
SIAM Journal of Mathematical Analysis 24(1993), pp.~407-435.

\bibitem{zhao} {\sc X.-Q. Zhao}, {\em Dynamical systems in population biology}, Springer, New York, 2003.

\end{thebibliography}
\end{document}